%% file: main.tex
\documentclass[11pt]{article}
\usepackage[margin=1in]{geometry}
\usepackage{booktabs}
\usepackage{mathtools}
\usepackage{amsthm}
\usepackage{amsmath}
\usepackage{amssymb}
\usepackage{tmpl}
\usepackage{booktabs} % For formal tables
\usepackage{mathtools}
\usepackage[ruled]{algorithm2e} % For algorithms

\SetAlFnt{\small}
\SetAlCapFnt{\small}
\SetAlCapNameFnt{\small}
\SetAlCapHSkip{0pt}
\IncMargin{-\parindent}

\hypersetup{
  colorlinks=true,
  linkcolor=blue,
  citecolor=blue,
  urlcolor=blue
}

\hypersetup{
  pdftitle={Improved Multi-Dimensional Forecasting for Swap Regret},
  pdfauthor={Joey Rivkin, Ramiro N. Deo-Campo Vuong,  Robert Kleinberg, Chido Onyeze, Erald Sinanaj, Eva Tardos}
}

\newcommand{\floor}[1]{\lfloor#1\rfloor}
\newcommand{\poly}{\textsc{Poly}}

\usepackage{natbib}
\setcitestyle{authoryear}

\title{Improved Multi-Dimensional Forecasting for Swap Regret}

\author{
\begin{tabular}{ccc}
Joey Rivkin & Ramiro N. Deo-Campo Vuong & Robert Kleinberg \\
\texttt{jbrivkin@cs.cornell.edu} &
\texttt{ramdcv@cs.cornell.edu} &
\texttt{rdk@cs.cornell.edu}
\\[1em]
Chido Onyeze  & Erald Sinanaj & Eva Tardos \\
\texttt{chidoonyeze@cs.cornell.edu} & \texttt{erald@cs.cornell.edu} &
\texttt{eva.tardos@cornell.edu}
\\[1em]
\multicolumn{3}{c}
{Cornell University}
\end{tabular}
}
\date{}

\begin{document}

\maketitle

\begin{abstract}
We study the problem of forecasting for an arbitrary number of downstream agents with unknown objectives, each of whom best responds to the forecaster's predictions. We seek a single forecaster that guarantees sublinear swap regret for all downstream agents simultaneously. For two-dimensional outcome spaces, we give a polynomial time algorithm that guarantees $\widetilde{O}(\sqrt{kT})$ swap regret for any downstream agent with $k$ actions. This improves over the previously known bound of $\widetilde{O}(kT^{5/8})$ and avoids the exponential in $T$ runtime of prior algorithms in this setting. Our algorithm extends nicely to other low dimensional environments, retaining $\widetilde{O}(\sqrt{T})$ downstream swap regret while the exponent of $k$ in the regret bound and the exponent of $T$ in the running time both grow with dimension. For arbitrary dimension $d$, we give a forecasting algorithm that guarantees $\widetilde{O}(d\sqrt{kT})$ swap regret, assuming the forecaster knows an upper bound $k$ on the number of actions available to any downstream agent, albeit with a much longer runtime. This improves upon previous high dimensional guarantees that had $\widetilde{O}(T^{2/3})$ dependence and required additional behavioral assumptions.
\end{abstract}

\renewcommand{\thefootnote}{}
\footnotetext{\scriptsize Robert Kleinberg, Ramiro N. Deo-Campo Vuong, and Erald Sinanaj were supported in part by NSF grant CCF-2402851. Chido Onyeze was supported in part by a SPROUT Award from Cornell Engineering. Eva Tardos was supported in part by AFOSR grants FA9550-23-1-0410 and FA9550-23-1-0068, and ONR MURI grant N000142412742.}
\renewcommand{\thefootnote}{\arabic{footnote}}
\newpage

%\tableofcontents

\input{1intro}
\input{2btechniquestake2}

\input{3related}
\input{4lowdims}
\input{3highdims}
\input{6lowerbounds}

\bibliographystyle{ACM-Reference-Format}
\bibliography{refs}

\input{7apdx}

\end{document}

%% file: 1intro.tex
\section{Introduction}

Predicting the state of the world is crucial for decision making. 
Weather forecasts inform daily plans, stock projections inform investment decisions, and traffic estimates inform route choice for commuters.

Consider a setting in which many individuals face \emph{sequential decision-making} problems. Each day, every individual must make a decision, and her payoff depends both on her decision and the uncertain state of the world. Many people lack the computational resources or data to make well-informed decisions on their own. This is where a public forecaster can help, producing a shared (and hopefully well-informed) forecast that many individuals use to guide their heterogeneous decisions. A forecaster is \emph{trustworthy} if individuals should feel confident treating its predictions as if they were the truth. But what does it mean to operationalize this conception of trustworthiness? How can a single forecaster fulfill the needs of many downstream agents simultaneously, each of 
whom could have entirely different decision tasks, action sets, and utilities?

While prior work in this widely studied area evaluates predictions primarily through statistical criteria such as proper scoring rules or calibration error, a growing body of literature assesses forecasts directly in terms of their decision-theoretic value for downstream agents. One such decision-theoretic aim, with game-theoretic foundations, is to ensure that any downstream agent who best responds to the forecasts will have low Swap Regret. In rough words, a commuter with low swap regret who ``looks back'' at all the days she chose to take route A would not wish she had instead taken some other route B on those same days: 
any such choice would have led to at least as much time spent waiting in traffic.

There are many reasons  to seek low Swap Regret; we highlight two. 
First, Swap Regret strengthens the well-known and widely used property of low External Regret, which has been shown empirically to accord with decision-makers' behavior in certain repeated decision-making environments \cite{nekipelov_econometrics,noti2021bid}. 

Second, a recent line of work shows that low swap-regret shields against potential manipulation by limiting an opponent's ability to extract excess value from them, even if the adversary knows the algorithm that the agent is running.

Algorithms for sublinear Swap Regret have been known since the seminal work of \citet{blum2007external}. 

However, all of these algorithms require a degree of sophistication and access to data that might not be available to a self-optimizing agent in many environments.
For instance, in our running example, the commuter might travel to a new country that she knows very little about.
These considerations motivate the question of whether a public forecaster can produce 
a sequence of predictions that, when combined with the agent's best-response behavior, lead to
Swap Regret rates comparable to Blum-Mansour. In the next sections we address this question using a model
introduced by \citet{rothshi}.

\subsection{Model}
We consider the following model of \cite{rothshi}: each day $t = 1, \dots, T$, an adaptive adversary sets the ``state of the world'' $y_t \in \cC$, where $\cC \subseteq \mathbb{R}^d$ is a convex, compact prediction space. The public forecaster must post a prediction $\hat{y}_t \in \cC$ before observing the true state. A collection of downstream agents, whose utilities depend on the unknown state $y_t$ and their chosen action $a \in \cA$, best-respond to this forecast, treating $\hat{y}_t$ as the true state. 

Formally, each agent has a utility function $u: \cA \times \cC \to [0,1]$. The \emph{Best Response} of an agent with utility $u$ is the action \[BR(u, \hat{y}) \coloneqq \arg \max _{a\in \cA} u(a, \hat{y}).\] where ties are broken consistently using a fixed ordering of the actions. Given a prediction $\hat{y}$ by the forecaster, the agent will take action $BR(u, \hat{y})$.

The adversary may be adaptive (and randomized): $y_t$ may depend arbitrarily on past transcript $\pi_{t-1} = (y_{1:t-1}, \hat{y}_{1:t-1})$. We assume that the utilities are affine in the outcome: for every action $a \in \cA$, the map $y \mapsto u(a,y)$ is affine on $\cC$. Note that this model is only more general than the setting in which there are finitely many possible states of the world $y$, like whether or not it will rain, forecasts represent probabilities over the discrete outcomes, and agents seek to maximize their expected utility.

We will denote with $L$ the $\ell_\infty$-Lipschitz constant of a utility function (with respect to the outcome). Since the utility is bounded in $[0,1]$, when $\cC$ is the unit cube, $L\le 1$. For more general convex regions, this Lipschitz constant is upper bounded by $d\sqrt{d}$ without loss of generality (for more discussion the reader is directed to the Appendix~\ref{lipschitz}), a fact which we will make use of in some of our results.

Our goal is to design forecasting algorithms that, with high probability, imply upper bounds on expected
swap regret simultaneously for every possible downstream agent. In fact, the analysis of our forecasting algorithms implies a slightly stronger result: a bound on the expectation of the maximum swap regret across all agents with a given number of actions in their decision problem. We now formalize these notions, first defining the downstream swap regret of a particular agent.

\begin{definition}[Downstream Swap Regret]
    Fix a transcript $\pi_T =(y_{1:T}, \hat{y}_{1:T})$. The \emph{Downstream Swap Regret} of an agent with utility $u: \cA\times \cC \to [0,1]$ is
    \[DSR_{\pi_T}(u) = \max_{\phi: \cA \to \cA}\sum_{t=1}^T\l[u\l(\phi(\br(u, \hat{y}_t)), y_t\r) - u(\br(u, \hat{y}_t), y_t)\r].\]
\end{definition}

We extend this notion to the worst-case over all downstream agents with some number of actions below.

\begin{definition}[Maximum Downstream Swap Regret]
    Define the \emph{Maximum Downstream Swap Regret} of a transcript $\pi_T$ across agents with $k$ actions as:
    \[MDSR_{\pi_T}(k) = \max_{u : [k]\times \cC \to [0,1]}DSR_{\pi_T}(u)\]
\end{definition}

When $\pi_T$ is clear from context, we omit transcripts and write $DSR(u)$ and $MDSR(k)$ respectively. Without loss of generality we can assume any action set $\cA$ of size $k = |\cA|$ is relabeled as $[k]$, maintaining the fixed ordering of actions. Further, note that any agent with fewer than $k$ actions can be modeled by an equivalent agent with $k$ actions by adding ``dummy'' actions to her set with utility 0 so $MDSR(k)$ captures all downstream agents with \emph{at most} $k$ actions.

\subsection{Our Results}
In \Cref{sec:4lowdims}, we provide an efficient algorithm that attains the following guarantee for forecasting in low (constant) dimensions.
\begin{theorem} [Informal version of \Cref{thm:lowformal}]\label{thm:intro-lowdim}
    Consider a prediction space $\cC \subseteq [0,1]^d$. There is a forecasting algorithm with the following guarantee: with probability at least $1 - \delta$ over the randomness of the forecaster, for every $k$
    \[MDSR(k) \leq C_d k^\frac{\lceil d/2\rceil}{2}\sqrt{T\ln (T / \delta)}\]
    where $C_d$ is a constant depending only on the dimension $d$.
    Moreover, for a fixed constant $d$ the forecasts can be made with per-round running time polynomial in $T$. 
\end{theorem}

For $d = 2$, the algorithm runs in polynomial time and guarantees that with high probability, every agent has swap regret at most $O(\sqrt{kT\ln T})$, matching the bound of \cite{blum2007external} (designed for a single agent) up to log factors. This improves upon the previously best known bound of $O(kT^{5/8}\sqrt{\ln(T)})$ which required running time exponential in $T$ as in \cite{rothshi}. Moreover, our guarantee gives a bound on the expectation of the maximum swap regret among downstream agents, a stronger guarantee than one on the maximum of their expected downstream swap regret. 

In higher dimensions, we show how to guarantee near optimal swap regret which scales only linearly in the dimension, when the forecaster has knowledge of an upper bound $k$ on the number of actions available to downstream agents, as the following theorem from \Cref{sec:3highdims} outlines.

\begin{theorem}[Informal version of \Cref{thm:high-dimensional}]\label{thm:intro-highdim}
    Consider a prediction space $\cC \subseteq [0,1]^d$. There is a forecasting algorithm with the following guarantee: given $k$ as an input, the algorithm produces forecasts such that with probability at least $1 - \frac{1}{T^{2d}}$, 
    \[MDSR_{\pi_T}(k) = O\l(d L\sqrt{kT\ln(T)}\r).\]
\end{theorem}

This gives the first $\tilde{O}(\sqrt{T})$ dependence on the horizon beyond the one-dimensional setting for agents who directly best respond. Prior work by \cite{rothshi} gave a  $O(T^{2/3})$ bound that relied on an added behavioral assumption that downstream agents no longer directly best respond but \emph{smoothly} respond (choosing actions randomly proportional to their utility). 

The Downstream Swap Regret bounds in \Cref{thm:intro-lowdim,thm:intro-highdim}
depend on the size of the agent's action set, $k$. It is natural to wonder if this 
dependence can be avoided, especially in light of the results of \citet{CDL-huWu}, which
show that no dependence on $k$ is necessary in dimension $d=1$.
In \Cref{sec:5lowerbounds}, we show that the dependence of Downstream Swap Regret
on $k$ in higher dimensions is probably unavoidable, by linking this question with
the theory of forecasting algorithms that minimize \emph{expected calibration error}, 
a topic already known to have strong connections to Downstream Swap Regret. We
exhibit a particular downstream agent, with $\Omega(T)$ actions, 
whose swap regret is quadratically related to the expected calibration error of the predictions. 
As a consequence of this construction, we show that bounding Maximum
Downstream Swap Regret by a function of the form $O_d(T^{1-\delta})$ for any
$\delta>0$, with the implied constant in the $O_d(\cdot)$ depending on dimension
but not on the agent's number of actions, would imply breakthrough results 
about expected calibration error in sequential forecasting.

%% file: 2btechniquestake2.tex
\subsection{Background}
We follow the framework of \cite{rothshi} of minimizing downstream swap regret by controlling the bias conditioned on the sets of predictions that activate a particular best response. We refer to these sets as best response regions.

\begin{definition}[Best Response Region] Fix an action set $\cA$ and utility function $u$. For any action $a \in \cA$,
we define its corresponding best response region to be $S_{u,a} = \{\hat{y} \in \cC\mid BR(u, \hat{y}) = a\}$. 
\end{definition}

We define the bias for any region $S$.

\begin{definition}[Subset Conditional Bias]
    For a subset $S \subseteq C$ we define the \emph{conditional bias} conditioned on that subset to be \[\beta_{\pi_T}(S)=\left\|\sum_{t=1}^T 1[\hat{y}_t \in S]\,(\hat y_t - y_t) \right\|_\infty .\]
\end{definition}
This definition is essentially identical to that of \cite{noarov2023high}, but we do not divide by $T$.

As observed by \cite{noarov2023high, rothshi}, a downstream agent's swap regret is controlled by the bias accumulated on best response regions. Intuitively, consider the days on which the agent chooses some action $a$ (i.e. $\hat{y}_t \in S_{u, a}$). On those days, the forecasted utility of $a$ is greater than that of any other action. If the realized outcomes, on average, are near the forecasted outcomes on the same subset of days (i.e. the conditional bias is small), then the forecasted utilities must be as well due to our assumption that the utility function is Lipschitz. In particular, this implies that action $a$ cannot be that much worse than any action on those days. This intuition is formalized in the following lemma:
\begin{lemma}\label{lem:sumbias}\cite{noarov2023high}
    Fix any agent with utility function $u: \cA \times \cC \to [0,1]$. Let $L$ be an upper bound for the $\ell_\infty$-Lipschitz constant for the second argument. Then
    \[DSR_{\pi_T}(u)  \leq 2L\sum_{a \in \cA} \beta(S_{u,a})\]
\end{lemma}

This provides a framework for controlling downstream swap regret: maintain small conditional bias on all best response regions. In order to obtain any meaningful guarantees, the forecaster should make predictions in some predetermined finite subset of the space, in order to control the bias of the induced best response sets (which are finitely many, in contrast to the entire space of best response partitions). Thus, as in \cite{rothshi, noarov2023high, CDL-huWu} we discretize the prediction space into a finite set. 
\begin{definition}[Discretized Prediction Set]
    For a prediction space $\cC \subseteq [0,1]^d$, the \emph{Discretized Prediction Set} $\cC_\eps \subset \cC$ is an $\ell_\infty$-net with resolution $\eps$. By taking points in a grid we can ensure $|\cC_\eps| \leq (1/\eps)^d$. We always take $\eps$ to be $1 / \poly(T)$ so $|\cC_\eps| = T^{O(d)}.$ 
\end{definition}

Prior works in this setting \cite{rothshi, noarov2023high} enumerate all subsets of $\cC_\eps$ that can arise as a best response region (of some utility), and use multi-objective optimization to simultaneously guarantee low bias on all such subsets. One challenge with this approach is that, in multiple dimensions, the number of induced subsets can be extremely large, posing issues both for downstream regret bounds and computational efficiency.

\subsection{Our Approach}
Our forecasting algorithms, too, use multi-objective optimization to control bias. 
The difference is that we reformulate the multiple objectives in the optimization 
problem in ways that lead to efficiency gains. The ways these efficiency
gains are achieved in low dimensions versus high dimensions are entirely opposite
one another: in low dimensions the key is to reduce the number of objectives considered, 
while in high dimensions the key is, surprisingly, to increase it.

In low dimensions, we control bias on a much simpler family of subsets (triangles or simplices) and show that best response regions can be written as a union of a bounded number of these building blocks. This leads to a computationally efficient forecasting algorithm, while also improving downstream swap regret guarantees. 

In high dimensions, the geometry of best response regions becomes more complex and the number of simplices needed to compose a best response region grows exponentially in the dimension $d$. This necessitates a different approach. Rather than breaking down best response regions, we enumerate all possible partitions of the discretized prediction space into best response regions, and, for each partition, control the total conditional bias aggregate across all best response regions in that partition. A central component of our analysis is showing that, although the space of utilities is infinite, the number of distinct partitions that arise is nevertheless well controlled.

Strikingly, in the $d=2$ case, these two seemingly opposite approaches lead to the same swap regret bound of $\tilde{O}(\sqrt{kT})$, both improving on previous bounds. A more detailed overview of the two approaches follows.

\subsubsection{Low Dimensions -- Triangulations}
The forecasting algorithm of \cite{rothshi} for dimensions $d = 1, 2$ relies on the observation that best response regions are convex. The forecasting algorithm given here relies on a tighter characterization of best response regions as convex polytopes with a number of faces bounded by the size of the action set, allowing us to control bias only on a much smaller family of sets.

Consider the case of $d = 2$. The best response region of an action $a$ is defined by linear comparisons $u(a, \hat{y}) \geq u(b, \hat{y})$ for all other actions $b$, and therefore is a polygon with at most one side contributed by each other action (and 4 additional sides from the boundary of $[0,1]^2$). With $k$ actions, each best response region is a polygon with at most $k+4$ sides and so can be decomposed into at most $k+2$ triangles. Therefore, the total cumulative bias summed across $k$ best response regions is bounded by a sum over $O(k^2)$ disjoint triangles, ($O(k)$ for each best response region). 

A more careful global count improves this to $O(k)$. Consider a graph whose nodes are best response regions and whose edges are the boundaries between them. This graph is planar because prediction space is two dimensional, and there are $k$ nodes, one for each action. Boundaries of best response regions correspond to edges of this graph. A planar graph with $k$ nodes can only have $O(k)$ edges, implying the same asymptotic bound on the total number of triangles needed. In other words, we can control the cumulative bias over the $k$ best response regions by controlling bias on only $O(k)$ many triangles. 

\begin{figure}[htp]
    \centering
    \includegraphics[width=9cm]{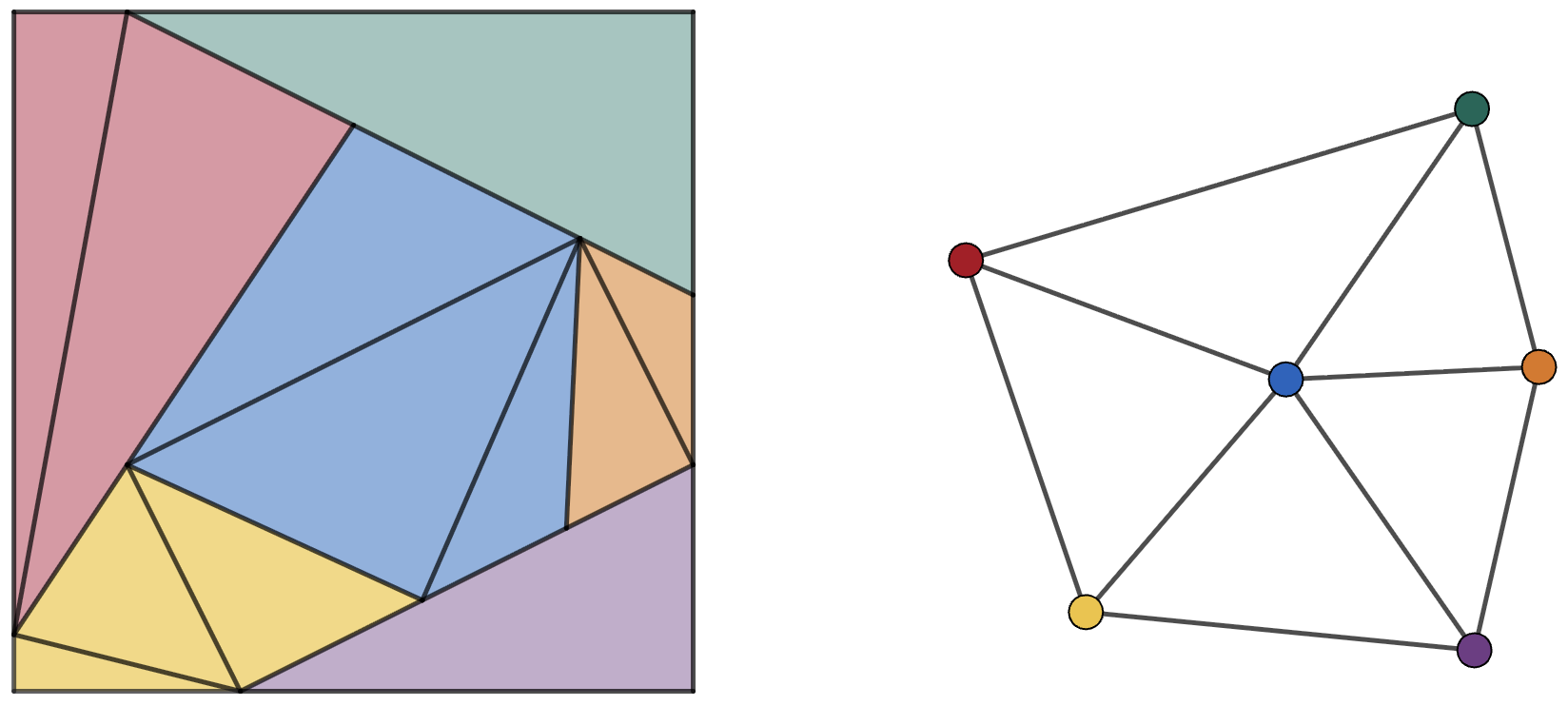}
    \caption{Best-response regions partitioning a prediction space (with the dual graph of the partition depicted at right) can be further partitioned into triangular sub-regions.}
\end{figure}

Beyond two dimensions, the geometry is more complex and it's not necessarily the case that the number of simplices needed to triangulate a polytope is linear in the number of its facets. However, using the Upper Bound Theorem from convex geometry, we are able to show that we do not need more than $O(k^{\lceil{d/2}\rceil})$ total simplices to triangulate the best response regions, where $d$ is the dimension (see \Cref{triangulationlemmas}).

Note that the forecaster is aiming to control bias (and hence swap regret) for many downstream agents without knowing the agents' utilities or best response regions. Therefore, our goal is to control the bias on all possible triangles of the forecast region. When the forecaster restricts to a discretized set of predictions, the number of subsets that arise from triangles or simplices is polynomial in the number of points, in contrast to the number of total convex sets which grows exponentially. This allows the forecaster to efficiently enumerate all such subsets and achieve improved bounds. 

The formal algorithm and analysis appear in \Cref{sec:4lowdims}.

\subsubsection{High Dimensions -- Counting Partitions}
The approach of \cite{rothshi} for forecasting in low dimensions was to maintain unbiasedness on every convex set. However, we can do better. Observe that an action $a$ is best response to forecast $\hat y$ if $u(a,\hat y)\ge u(b,\hat y)$ for all other actions $b$. For any given action $b$, this is a linear constraint in $\hat y$, implying that best response regions are the intersections of at most $k-1$ half spaces. 
By maintaining conditional unbiasedness on subsets of $\cC_\eps$ realizable by such intersections, we can achieve swap regret $O(Lkd\sqrt{T \ln (T)})$ for any downstream agent with at most $k$ actions. This already strengthens prior high dimensional guarantees (See Appendix~\ref{OldHighDim}).
However, this bound incurs a $\sqrt{k}$ factor greater regret than the $O(\sqrt{kT})$ rate achievable by a single agent running a no swap regret algorithm. 

Intuitively, this linear dependence on $k$ comes from summing bias across the $k$ best response regions, and swap-regret is the total regret swapping each action to a different action. Motivated by this difficulty, we shift perspective and instead control the total bias at the level of entire best response \emph{partitions}.
Specifically, the forecaster enumerates every possible partition of $\cC_\eps$ into best response regions that could arise from a utility function, with the goal being to control the total bias, summed across best response regions, for every such partition. 
The main hurdle to this approach is showing that there are not too many possible partitions for the forecaster to deal with. 
When the domain is restricted to the discretized prediction set, there are still infinitely many utility functions, but there are only a finite number of possible partitions. 
Naively, that number of partitions could be extremely large (e.g. $k^{|\cC_\eps|} = k^{T^d}$), but we show that the number of partitions scales roughly as $|\cC_\eps|^{kd} = T^{kd^2}$. We do this by casting utilities on $k$ actions as elements in parameter space $\R^{k(d+1)}$. Then, for each point in $\hat{y} \in \cC_\eps$ and pair of actions $a, b$ we consider the hyperplane on which $u(\hat{y}, a) = u(\hat{y}, b)$. Distinct partitions correspond to distinct cells cut out by these hyperplanes, so the problem is reduced to counting such cells.

Finally, to simultaneously control the total bias on each partition, we reduce to multiobjective optimization. We show how to write bias on one partition as the maximum of
$(2d)^k$ linear functions 
and show that it is sufficient to consider only $\approx T^{kd^2}$ partitions.  So controlling bias on all partitions requires controlling as many as $N\approx (2d)^kT^{kd^2}$ linear objectives. Standard multiobjective optimization guarantees scale as $\sqrt{T \ln(N)}$ where $N$ is the number of objectives. For $N \approx (2d)^kT^{kd^2}$, the second factor dominates and this gives the $O(d \sqrt{kT \ln(T)})$ downstream swap regret bound.

The formal algorithm and analysis appear in \Cref{sec:3highdims}.

%% file: 3related.tex
\subsection{Related Work}

One of the most well known and now widely studied metrics of prediction accuracy is Expected Calibration Error (ECE), also known as $\ell_1$ Calibration Error. 
In short, the classical definition measures the $\ell_1$-distance of the vector of unique predictions made by the forecaster to the ``calibrated'' vector: 
where each unique prediction is swapped to the empirical observed state. 
It has been known for quite some time, due to \citet{Calibration-fosterVohra}, that a forecaster with low ECE guarantees low \emph{Swap Regret} to any agent that best responds to her predictions. 
Their algorithm and its pessimistic guarantee, growing as $O(T^\frac{d+1}{d+2})$, remained the best we know to date for non-constant number of labels, $d+1$, and were only very recently surpassed by the work of \citet{breakingCalibration-DDFGKO} for the two-label setting, which gives a guarantee of $O(T^{2/3 - \veps})$ for some small $\veps$. 
Recently, the work of \citet{highDimCal-peng} gives an algorithm that achieves 
$\ell_1$ Calibration Error less than $\epsilon T$ provided $T = \Omega ( d^{1/\veps^2} )$,
avoiding the exponential dependence on $d$ in the $T = \Omega \l( (1/\veps)^{d+2} \r)$
bound implied by \citet{Calibration-fosterVohra}, while substituting an exponential 
dependence on $1/\veps$.
Subsequently, \citet{treeSwap-FGMS} generalize this result to arbitrary norms other than $\ell_1$ and give a stronger lower bound for $\ell_1$ Calibration Error.  

Minimizing Expected Calibration Error is statistically costly, suffering
from a lower bound of $\Omega(T^{0.54})$ \citep{qiao2021stronger,breakingCalibration-DDFGKO}, 
which constitutes a barrier to obtaining $\tilde{O}(\sqrt{T})$ Downstream Swap Regret 
(or even Downstream Expected Regret) via forecasters with low ECE. 
However, it is also known that low ECE is not \emph{necessary} for low downstream Swap Regret and various recent works provide downstream guarantees that are not achievable by standard calibration techniques. 
The work of \citet{Ucal-kleinberg} investigates downstream no External Regret and provides a novel framework and view through \emph{proper scoring rules}, 
defining \emph{U-calibration} as an efficiently attainable property to guarantee an optimal $O(\sqrt{T})$ external regret for any downstream decision maker in the one-dimensional case.
They also provide an efficient algorithm giving a downstream expected regret
guarantee of $O(d\sqrt{T})$, 
which was later improved by \citet{OptimalUcal-luoSenapatiSharan} to the optimal $O(\sqrt{dT})$.
These guarantees in $d>1$ dimensions, unlike the $d=1$ guarantee
from \citep{Ucal-kleinberg}, are for the weaker objective of maximum
expected Downstream External Regret. To the best of our knowledge, for the 
stronger objective of \emph{expected maximum} Downstream External Regret
in dimension $d>1$, no bounds asymptotically outperforming the bounds
arising from ECE minimization were known prior to our work.

In regards to the harder guarantee of No Swap Regret, a number of works also give positive results albeit in limited settings. 
First, \citet{rothshi}, largely based on the framework set by \citet{noarov2023high}, were able to provide an efficient forecasting algorithm that guarantees $\tilde{O}(k\sqrt{T})$ Swap Regret for any downstream agent with $k$ actions in the one-dimensional setting. 
They extended this idea to two dimensions giving a guarantee that scales as $\tilde{O}(kT^{5/8})$, 
but with an algorithm that runs in time exponential in $T^{\alpha}$ for some $\alpha>0$. 
More recently, \citet{CDL-huWu} use the ideas introduced in \cite{li2022optimization,Ucal-kleinberg} and \cite{noarov2023high} to give a near-optimal action-independent bound of $O(\sqrt{T\log T})$ in the one-dimensional setting. 

For a single agent with $k$ actions minimizing Swap Regret (as opposed to a forecaster minimizing
Downstream Swap Regret for all $k$-action agents simultaneously) the canonical algorithm,
due to \citet{blum2007external}, achieves a Swap Regret bound of $O(\sqrt{T k \log k})$.
A matching lower bound (up to constant factors) was presented by \citet{ito2020tight},
improving an earlier $\Omega(\sqrt{Tk})$ lower bound of \citet{blum2007external}.
Remarkably, our high-dimensional forecasting algorithm (\Cref{thm:high-dimensional})
achieves a Maximum Downstream Swap Regret bound of order $O(\sqrt{T k \log T})$
(omitting factors that depend on the dimension and Lipschitz constant), 
matching the Blum-Mansour bound up to the $\log(T)$ factor.
Our bounds are incomparable with the Swap Regret bounds of 
\cite{peng2024fast,dagan2024external}, which have
polylogarithmic dependence on the number of actions but 
$T^{1-o(1)}$ dependence on the time horizon.

%% file: 4lowdims.tex
\section{Efficient Forecasting In Low Dimensions}\label{sec:4lowdims}
A central barrier to computationally efficient forecasting is the large number of possible best response regions, even when restricted to the discrete prediction set $\cC_\eps$. Algorithms that explicitly enumerate all such regions necessarily incur proportionally high runtime. We avoid this blowup by replacing the (large) family of best response regions with the much smaller family of simplex-induced subsets on the discretized prediction space as defined below.

\begin{definition}
    Define the set of simplex-induced subsets of the discretized prediction space to be \[\cS_\eps \coloneq \l\{\cC_\eps \cap \bigcap_{i=1}^{d+1} H_i \mid H_i\text{ a halfspace in }\R^d\r\}. \]
\end{definition}

We show a modest bound on the number of such induced subsets of $\cC_\eps$ using VC dimension analysis, and show that any agent's best response regions can be composed into a well-bounded number of such subsets. To maintain unbiasedness on all of these subsets, we instantiate the prediction algorithm of \cite{noarov2023high} with the collection of these simplex-induced subsets of $\cC_\eps$. We state here the main guarantee of \cite{noarov2023high}.

\begin{theorem} [\cite{noarov2023high}] \label{thm:nor}
    Fix a convex outcome space $\cC \subseteq [0,1]^d$, discretized prediction set $\cC_\eps \subset \cC$, and collection of subsets $\cE \subset 2^{\cC_\eps}$. There is an algorithm producing $p_1, \ldots, p_T \in \Delta \cC_\eps$ such that for any sequence of outcomes $y_1, \ldots, y_T \in \cC$ chosen by the adversary:
    \[
    \mathbb{E}_{\hat{y}_t' \sim p_t \forall t}
    \l\|\sum_{t=1}^T1[\hat{y}'_t \in S](\hat{y}'_t - y_t)\r\|_\infty \leq O\left(\ln(d|\cE|T) + \sqrt{\ln(d|\cE|T) |n_T(S)|} + \eps T\right)
    \]
    where $n_T(S) = \sum_{t=1}^T  \mathbb{E}_{\hat{y}'_t \sim p_t} 1[\hat{y}'_t \in S]$. Additionally, the algorithm can be implemented with per-round running time scaling polynomially in $d$, $|\cE|$, and $T$.
\end{theorem}

Using the algorithm of \cite{noarov2023high}, a forecaster can be approximately unbiased on the collection of simplex-induced subsets, incurring small bias for each one that depends on the frequency ($|n_T(S)|$) of predicting from the set $S$  and on the size of the collection $|\cE|$. We now state the main result of the section: a forecasting algorithm that is oblivious to the downstream agents using it and which guarantees a high probability bound on the worst case downstream swap regret, as a function of the number of actions of each agent (which is unknown to the forecaster), the dimension, and the horizon:

\begin{theorem}\label{thm:lowformal}
    Consider a prediction space $\cC \subseteq [0,1]^d$ and let $\cC_\eps$ be an $\ell_\infty$ net with resolution $\eps$. Consider the \emph{unbiased predictions} algorithm of \cite{noarov2023high} instantiated with $\mathcal{E} = \{E_\sigma\}_{\sigma \in \cS_\epsilon}$ and $\epsilon = \frac{1}{T^2}$. Then, for any adversary (possibly adaptive and randomized), with probability at least $1 - \delta$ the following holds simultaneously for all $k$:  \[MDSR(k) \leq C_dk^{\frac{\lceil d / 2\rceil}{2}}\sqrt{T\ln(T/\delta)}\] 
    where $C_d$ is a constant depending only on $d$. Moreover, the forecasts can be made with per-round running time polynomial in $T^{d^3}$.
\end{theorem}

Observe that for $d=2$ this gives a bound of $O(\sqrt{kT\ln(T)})$ on expected maximum case swap regret and a $O(k\sqrt{T\ln(T)})$ bound for dimensions $d=3,4$ (by setting say $\delta = 1/T^2$). The first matches well known bounds for an agent optimizing for swap regret themselves, and all three, to our knowledge, improve on any previously known algorithms for downstream swap regret. Moreover, the algorithm is efficient (polynomial in $T$, for a fixed dimension $d$). 

The remainder of this section works toward proving Theorem~\ref{thm:lowformal}. In subsection~\ref{sec:2.2/geometry_and_intuition}, we bound the number of simplex-induced subsets and the number of such subsets needed to compose best response regions. Then, in subsection~\ref{sec:2.1/highprob}, we show that the algorithm of \cite{noarov2023high} provides a strong high probability bound on the total bias across disjoint events. Finally, in subsection~\ref{sec:2.3/bound_proof}, we prove the main result.

\subsection{Geometric tools to analyze the algorithm}\label{sec:2.2/geometry_and_intuition}

To analyze the algorithm we must show that (1) the family $\cS_\eps$ is not too large and (2) for any agent, each best response region can be broken down into a reasonable number of sets from $\cS_\eps$. The two ideas are captured by the following two lemmas. The first bounds the total number of subsets of the discretized prediction set induced by simplices, and the second shows that any agent's best response regions can be composed by a bounded number of these simplex-induced subsets.

\begin{restatable}{lemma}{simplex_count}\label{lem:simplex_count}
Let $P \subset \R^d$ be a finite set with $|P| = m$, and define: \[\mathcal{P}^{(n)} = \left\{P \cap \bigcap_{j=1}^n H_i \mid H_1, \ldots, H_n \text{ are halfspaces in }\R^d\right\}.\] Then the size of the collection is bounded as $|\mathcal{P}^{(n)}| = m^{O(nd)}$. 
\end{restatable}

Setting $n = d+1$ gives $|\cC_\eps|^{O(d^2)}$ bound on the number of simplex-induced subsets. The proof of Lemma~\ref{lem:simplex_count} proceeds by first using VC dimension analysis to bound the number of sets of discretized prediction values induced by half-spaces and then bounding the number of $n$-wise intersections. We rely on the following two facts about VC dimension:

\begin{fact}[see e.g. \cite{vapnik1998statistical}]
    The VC dimension of the family of halfspaces in $\mathbb{R}^d$ is $d + 1$.
\end{fact}
\begin{fact}[Sauer-Shelah lemma, \cite{sauer1972, shelah1972}] 
Let $\mathcal{F}$ be a family of subsets of an $m$-point set with VC-dimension $D$. Then $|\mathcal{F}| = O(m^D)$. 
\end{fact}
\begin{proof}[Proof of \Cref{lem:simplex_count}]
    Let \[\cF := \{P \cap H : H \text{ is a halfspace in } \mathbb{R}^d\}.\] By the Sauer-Shelah lemma and the fact that the VC-dimension of halfspaces in $\R^d$ is $d+1$, we have $|\cF| = O(m^{d+1})$. Note that every set in $\mathcal{P}^{(n)}$ can be written as the intersection of $n$ sets from $\cF$: \[P \cap \bigcap_{j=1}^n H_j = \bigcap_{j=1}^n (P \cap H_j).\]
    Therefore, \[|\cP^{(n)}| \leq |\cF|^n = O(m^{n(d+1)}) = m^{O(nd)}.\]
\end{proof}

The other core lemma bounds the number of simplices needed to compose the best response regions of an agent.

\begin{restatable}{lemma}{triangulation}\label{lem:triangulation}
    Consider a convex, compact space $\cC \subseteq [0,1]^d$ and a finite collection of points $P\subseteq \cC$. Fix a utility function $u : [k]\times \cC \to [0,1]$, affine in its second argument. There exists a collection $\cS$ of simplices in $\R^d$ with disjoint interiors such that
    \begin{enumerate}
        \item $|\cS|\leq c_dk^{\lceil d/2 \rceil}$ for some constant depending only on the dimension $d$
        \item no point of $P$ lies on the boundary of a simplex $\sigma \in \cS$
        \item every point in $P$ lies in the interior of exactly one simplex $\sigma \in \cS$
        \item for every simplex $\sigma \in \cS$, the set $P \cap \sigma$ contains points from at most one best response region of $u$
    \end{enumerate}
\end{restatable}

The full proof of the lemma can be found in Appendix~\ref{triangulationlemmas}, but we sketch the main ideas here. Recall that for agents with $k$ actions, best response regions are convex polytopes in $\R^d$ with $O(k)$ facets. The Upper Bound Theorem \cite{McMullen_1970} from convex geometry implies that a simple convex polytope with $k$ facets can be triangulated into $O\l(k^{\lfloor d / 2 \rfloor}\r)$ simplices. Applying this to each of the $k$ best response regions gives an $O(k \cdot k^{\lfloor d/2\rfloor})$ upper bound on the number of simplex-induced subsets needed. 

To improve on this naive count, we represent each affine utility $u(a, \cdot)$ as a hyperplane in $\R^{d+1}$ (via standard lifting). Best response regions correspond to facets of the upper envelope. Rather than triangulating best responses separately, we triangulate the full upper envelope which induces a triangulation of the facets (and therefore the best response regions). Because the upper envelope is a $d+1$-dimensional polytope with $O(k)$ facets, it can be triangulated using $O\l(k^{\lceil d/2\rceil}\r)$ simplices, an improved bound in even dimensions. 

The main additional subtlety that arises is the case of points in the discretized prediction set landing exactly on the boundaries of simplices. However, by translation by a uniform random vector in a ball, this occurs with probability 0. Therefore, by the probabilistic method, there exists some decomposition into simplices for which no points in the discretized prediction set lie on simplex boundaries.

\subsection{A high probability bound on the unbiasedness on disjoint subsets}\label{sec:2.1/highprob}
We seek high probability bounds on maximum downstream swap regret across all possible agents. 
This implies bounds on the \emph{expected maximum downstream swap regret}, which is stronger than typical bounds on the \emph{maximum expected downstream swap regret} across downstream agents. 

The Theorem from \cite{noarov2023high} only bounds expected bias with respect to draws from the mixtures over predictions $p_t \in \Delta\cC_\eps$. 
For an adaptive adversary, who might respond to the realized draws $\hat{y_t} \sim p_t$, this does not immediately give a bound on the realized bias. However, by considering a martingale which tracks the difference between the realized bias (with respect to the realized forecasts $\hat{y}_t$) and expected bias (with respect to the randomness of draws from $p_t$), we show that the algorithm of \cite{noarov2023high} indeed provides high probability guarantees with respect to the realized forecasts.

\begin{restatable}{lemma}{uniformwhp}\label{thm:uniformwhp}
    Let $\cC \subseteq [0,1]^d$ be a convex prediction--outcome space and let $\cE$ be a finite collection of subsets. For any (possibly adaptive and randomized) adversary that produces outcomes $y_t \in \cC$ as a function of the past forecasts and outcomes $\pi_{t-1}$, the algorithm of \cite{noarov2023high} draws predictions $\hat y_1,\ldots,\hat y_T \in \cC_\eps$ from $p_1,\ldots, p_T\in \Delta\cC_\eps$ such that with probability at least $1 - \delta$, for every subset $S \in \cE$:
    
    \[\left\|\sum_{t = 1}^T 1[\hat{y}_t \in S](\hat{y}_t - y_t)\right\|_\infty \leq c\left(\ln(d|\cE|T/\delta) + \sqrt{\ln(d|\cE|T/\delta)n_T(S)} + \eps T\right),\]
    where the probability is over the randomness of both the adversary and the forecaster, $n_T(S) = \mathbb{E}_{\hat{y}'_t \sim p_t}\sum_{t=1}^T 1[\hat{y}'_t \in S]$, and $c$ is a universal constant. The algorithm can be implemented with per-round running time scaling polynomially in $d$, $|\cE|$, and $T$. 
\end{restatable}

The full proof, which uses Freedman's inequality \cite{3d1c033e-51cd-3f22-8e4e-611c9a4badb9} to bound the martingale, is left to the appendix (see \Cref{uniformwhp_proof}). The analysis of our algorithm relies on bounding the bias summed across disjoint simplex-induced subsets (which together compose an agent's best response regions). For disjoint subsets, the counts $n_T(S)$ sum to at most $T$; this observation along with Jensen's inequality yield the following.

\begin{restatable}{lemma}{jensen}\label{cor:jensen}
    Instantiating the algorithm with $\eps = 1/T^2$, with probability at least $1-\delta$, for any set of disjoint subsets $\cR \subseteq \cE$, 
    \[\sum_{S \in \cR} \beta(S) \leq c' \l(\sqrt{|\cR|T\ln(d |\cE| T /\delta)}\r)\]
    where $c'$ is a universal constant.
\end{restatable}
\begin{proof}
    Suppose the conditions of Lemma~\ref{thm:uniformwhp} hold. Then,
    \begin{align*}
        \sum_{S \in \cR} \beta(S) &\leq c \sum_{S \in \cR}\left(\ln(d|\cE|T/\delta) + \sqrt{\ln(d|\cE|T/\delta)n_T(S)} + \eps T\right)\\
        &= c \left(|\cR|\ln(d|\cE|T/\delta) + \sum_{S \in \cR}\sqrt{\ln(d|\cE|T/\delta)n_T(S)} + |\cR|/T\right) \\
        & \leq c \left(|\cR|\ln(d|\cE|T/\delta) + \sqrt{T\ln(d|\cE|T/\delta)|\cR|} + |\cR|/T\right)\quad\text{(Jensen's inequality)}
    \end{align*}
    Note that the total bias across disjoint events is always bounded above by $T$. Either the middle term $\sqrt{T\ln(d|\cE|T/\delta)|\cR|}$ is the largest or is itself greater than $T$. Therefore, setting $c' = 3\max\{1, c\}$ we get the desired guarantee.
\end{proof}

\subsection{Proof of Theorem~\ref{thm:lowformal}}\label{sec:2.3/bound_proof}
With the preceding Lemmas (\ref{lem:simplex_count}, \ref{lem:triangulation}, and \ref{cor:jensen}) the proof of \Cref{thm:lowformal} is straightforward.

\begin{proof}[Proof of Theorem~\ref{thm:lowformal}] By Lemma~\ref{lem:sumbias}, it suffices to bound the cumulative bias conditional on the best response regions $S_{u, a}, a \in [k]$ of the agent, i.e. the sum: $\sum_{a \in [k]} \beta(S_{u, a})$. Applying \Cref{lem:triangulation} gives a collection of disjoint simplex-induced sets $s_{u, a}^{(i)}, i \in [N_a]$ with $\bigsqcup_{i \in [N_a]} s_{u, a}^{(i)}$ containing exactly $S_{u, a} \cap \cC_\eps$ and $\sum_{a\in[k]} N_a \leq c_dk^{\ceil{d/2}}$ (note that any simplex in $\R^d$ can be written as the intersection of $d+1$ halfspaces). Using this decomposition and applying the triangle inequality to the biases,

\[\sum_{a \in \cA} \beta(S_{u, a}) = \sum_{a \in \cA} \beta\l(\bigsqcup_{i \in [N_a]} s_{u, a}^{(i)}\r) \leq \sum_{a \in \cA}\sum_{i \in [N_a]} \beta\l(s_{u, a}^{(i)}\r).\]

Let $|\cR| = \sum_{a \in [k]}N_a \leq c_dk^{\ceil{d/2}}$ denote the number of disjoint simplex-induced subsets $\cR$, that partition the best response regions of the agent. Since we make use of the \emph{unbiased predictions} algorithm of \cite{noarov2023high} with the collection of sets to be unbiased to, the set $\cE$ of all such simplex-induced subsets, we have the following bound, with probability at least $1 - \delta$ (from \Cref{cor:jensen}):

\[\sum_{a \in \cA}\sum_{i \in [N_a]} \beta\l(s_{u, a}^{(i)}\r) \leq c'\l(\sqrt{|\cR| T \ln\l(d|\cE|T/\delta\r)}\r)\]

Now, recall that \Cref{lem:simplex_count} bounds the cardinality of $\cE$. For our purposes we have $P = \cC_\eps$ and observe that $m = |P| \leq \l(\frac{1}{\eps}\r)^d = T^{2d}$. Thus, we can bound: \[|\cE| \leq T^{O(d^3)}.\] Substituting in these bounds above and absorbing constants depending only on $d$ into $C_d$ gives the stated bound, as the agent was arbitrary and the high-probability event holds uniformly over all simplex-induced subsets:

\[MDSR(k) \leq C_d k^\frac{\lceil d/2\rceil}{2}\sqrt{T\ln (T / \delta)}\]

Recall that the Lipschitz constant can be assumed to be polynomial in $d$ (specifically, $L \le d\sqrt{d}$, see discussion in \Cref{lipschitz} of the appendix), and so can be absorbed into $C_d$. The per-round runtime is polynomial in $|\cE|$, $T$, and $d$ and so polynomial in $T^{d^3}$.
\end{proof}

%% file: 3highdims.tex
\section{Forecasting In High Dimensions}\label{sec:3highdims}
An agent who knows their own utilities 
and runs a standard swap-regret algorithm (e.g. \cite{blum2007external}) can guarantee swap regret $\Tilde{O}(\sqrt{kT})$. In this section we show that a \emph{single} public forecaster can achieve essentially the same $\Tilde{O}(\sqrt{kT})$ rate \emph{simultaneously for all downstream agents} with at most $k$ actions, up to a linear dependence on the dimension of the prediction space. 

\begin{theorem} \label{thm:high-dimensional}
    Fix a prediction space $\cC \subseteq [0,1]^d$ and parameter $k \in \mathbb{N}$. There is an algorithm that produces forecasts $\hat{y_t}$ based on the transcript of past history $\pi_{t-1}$ such that for every (possibly adaptive and randomized) adversary mapping transcripts to outcomes $y_t$, with probability at least $1 - \frac{1}{T^{2d}}$, 
    \[MDSR_{\pi_T}(k) =O\l(d L\sqrt{kT\ln(T)}\r).\]
\end{theorem}

\subsection{Proof Overview}
Rather than enforcing unbiasedness separately on each best response region, our algorithm and analysis control bias at the level of an entire \emph{best-response partition}.  

Fix the discretized prediction set $\cC_\eps \subset \cC \subset [0,1]^d$. Every utility function $u$ induces a labeled partition of $\cC_\eps$ into best response regions, i.e. a map \[\rho: \cC_\eps \to [k], \quad \rho(\hat{y}) = BR(u, \hat{y}).\] Let $\cD_k(\cC_\eps)$ denote the set of all such realizable partitions. For a partition $\rho \in \cD_k(\cC_\eps)$, define its \emph{total bias} to be \[\beta_{\pi_T}(\rho) \coloneq \sum_{j \in [k]} \beta_{\pi_T}(\rho^{-1}(j)) = \sum_{j \in [k]} \l\| \sum_{t: \rho(\hat{y}_t) = j} (\hat{y}_t - y_t) \r\|_\infty\]
which (up to a factor of $2L$) bounds the downstream swap regret. We bound \emph{total bias} uniformly over \emph{every realizable partition} $\rho$ that could arise from some utility function.

We do this via the following steps:
\begin{enumerate}
    \item First, we bound the cardinality of the set of realizable partitions $|\cD_k(\cC_\eps)|$. We view utility functions as points in the parameter space $\R^{k(d+1)}$ and count the number of subsets of this space which induce different partitions. For every value $\hat{y} \in \cC_\eps$ and pair of actions $a,b \in [k]$ we consider a hyperplane within the parameter space defined by $u(a, \hat{y}) = u(b, \hat{y})$. We bound the total number of cells that are formed by these hyperplanes which corresponds to the number of partitions. 

    \item Next, we encode the \emph{total bias} of each partition $\rho$ as a maximum over $(2d)^k $ linear loss objectives, one for each possible assignment of signed coordinate directions to best-response regions in the best-response partition.
    \item Combining (1) and (2) gives a family of $D = (2d)^k|\cD_k(\cC_\eps)|$ loss objectives. We then apply the multiobjective optimization framework of \cite{calibeating} which ensures a high probability $O(\sqrt{T \ln(D)})$ bound on the maximum loss.
\end{enumerate}

\subsection{Bounding the number of realizable partitions}
Naively, partitioning $n$ points into $k$ labeled sets allows for as many as $k^n$ partitions. However, given the geometric structure of best response regions we can get a much tighter bound.

\begin{lemma}\label{lem:countpartitions}
    Fix $\cC_\eps \subset \cC \subset [0,1]^d$ with $n = |\cC_\eps|$. Then \[|\cD_k(\cC_\eps)| \leq \sum_{i=0}^{k(d+1)}\binom{\binom{k}{2} n}{i}  \leq (nk^2)^{k(d+1)}\]
\end{lemma}

\begin{proof}
    Each affine function $u(a,\cdot) : \cC\subset \R^d \to [0,1]$ can be specified by a coefficient vector in $\R^{d+1}$, thus an agent's utility function $u : [k] \times \cC \to [0,1]$ can be identified with a vector in $\R^{k(d+1)}$. For a point $y \in \cC_\eps$ and pair $a < b \in [k]$, let $H_{y, (a,b)} = \{u \mid u(a, y) = u(b,y)\}$ and note that this is a hyperplane in $\R^{k(d+1)}$ (when thinking of utilities $u$ as points in that space). The hyperplanes $\{H_{y, (a,b)}\}_{y \in \cC_\eps, a<b \in [k]}$ divide the space of utility functions into cells. All of the utility functions that correspond to a cell would induce the same partition of the points of $\cC_\eps$ into best response regions. Therefore $|\cD_k(\cC_\eps)|$ is at most the number of such cells. A standard bound states that $m$ hyperplanes divide $r$ dimensional space into at most $\sum_{i=0}^{r} \binom{m}{i}$
    cells (see e.g. \cite{stanley2007introduction}). We have $m = \binom{k}{2} n$ hyperplanes and $r = k(d+1)$ is the dimension of the utility space, which yields the claim. 
\end{proof}

\subsection{Encoding Total Bias as a Maximum of Loss functions}
We encode the total $\ell_\infty$ bias of a partition as a maximum over a finite family of linear loss objectives. A standard identity allows $\ell_\infty$-norms to be written as a maximum of $2d$ quantities, one for each signed direction: $\|x\|_\infty=\max_{i \in d} \max (x_i, -x_i)$. Using this, the sum of $k$ many $\ell_\infty$ norms can be written as the maximum of $(2d)^k$ many loss objectives.
 
The loss functions just defined are indexed by triples $(\rho, \sigma, I)$ where $\rho \in \cD_k(\cC_\eps)$, $\sigma \in \{\pm1\}^k$ and $I \in [d]^k$. Intuitively, $(\sigma, I)$ assigns each best response region $a\in[k]$ a coordinate direction $I_a$ and sign $\sigma_a$.

When the forecaster plays $\hat{y}_t \sim p_t$ and the adversary chooses $y_t \in \cC$, we define $a_t = \rho(\hat{y}_t)$ and set the loss of objective $(\rho, \sigma, I)$ to be
\[l_{\rho, \sigma, I}(\hat{y}_t, y_t) = \sigma_{a_t} \l(y_{t, I_{a_t}} - \hat{y}_{t, I_{a_t}}\r).\]
For a fixed partition $\rho$, the cumulative loss is,
\begin{align}\label{eq:loss}
    \sum_{t=1}^T l_{\rho, \sigma, I}(\hat{y}_t, y_t) = \sum_{a\in [k]} \sigma_a \sum_{t: a_t = a}\l(y_{t, I_a} - \hat{y}_{t, I_a}\r)
\end{align}
The maximum loss related to $\rho$ is exactly the total bias on that partition
\[
    \beta(\rho) = \sum_{a\in [k]}\l\|\sum_{t: a_t
    = a}\l(y_{t} - \hat{y}_{t}\r)\r\|_\infty
    = \max_{\sigma, I} \sum_{a\in [k]} \sigma_a \sum_{t: a_t = a}\l(y_{t, I_{a}} - \hat{y}_{t, I_a}\r)
\]
Hence, the maximum of the family of losses $(\rho, \sigma, I)$ is exactly the maximum total bias across all partitions.

\subsection{Applying Multiobjective Optimization}
To simultaneously bound all of the losses, we reduce to the multiobjective optimization algorithm of \cite{calibeating}. Their work offers an online algorithm that keeps the maximum loss across a vector of losses small.

Their model applies to a more general setting, which allows for time-varying action spaces and loss functions. A specialization of their model that suffices for our setting follows (See Appendix~\ref{calibeatingmodel} for explanation of how their more general model implies the theorem as stated here):

Consider a game between a learner and an adversary.
At every timestep $t=1, \ldots, T$, the learner chooses an action $x_t \in \cX$ and simultaneously, the adversary chooses an action $y_t \in \cY$. The learner incurs loss $l(x_t, y_t)$ according to a vector loss function $l(\cdot, \cdot): \cX \times \cY \to [-1, 1]^D$ with each $l_j(x, \cdot)$ concave in its second argument. Let \[w \coloneq \sup_{y \in \cY} \min_{x \in \cX} \max_{j \in [D]} l_j(x, y)\] i.e. the bound on the maximum loss that the learner incurs when the adversary is forced to move first and the learner can respond. The learning algorithm's goal is to ensure that over a time horizon of $T$ steps the learner does not suffer much more than $wT$ loss, the additional loss is the analog of classical regret in the case of a single loss function.

\begin{theorem}[Adapted from \cite{calibeating}, Theorem~A.2]\label{calibeating_main}
    Fix any $\delta \in (0,1)$. Given $T \geq \ln(D)$, there exists an algorithm that guarantees with probability $1 - \delta$ over the learner's randomness, 
\[\max_{j \in [D]}\l\{\sum_{t = 1}^T l_j(x_t, y_t)\r\} \leq 8 \sqrt{T \ln\l(\frac{D}{\delta}\r)} + wT\] 
\end{theorem} 

We instantiate this theorem with $\cX = \cC_\eps$, $\cY = \cC$, the losses in \Cref{eq:loss}, and $D = (2d)^k|\cD_\eps^{(k)}|$. 
We set $\eps = 1/T$ and $\delta = 1/T^{2d}$. Since $\cC_\eps$ is an $\ell_\infty$-net of resolution $\eps$, for any $y \in \cC$ there exists $\hat{y} \in \cC_\eps$ with $\|y - \hat{y}\|_\infty \leq \eps$, which implies that $w \leq \eps$.
Therefore, with probability at least $1 - \delta$, 
\[\max_{\rho \in \cD_\eps^{(k)}} \sum_{a \in [k]} \l\| \sum_{t: \rho(\hat{y}_t) = a} (\hat{y}_t - y_t) \r\|_\infty \leq 8\sqrt{T\ln\l(\frac{D}{\delta}\r)} + 1\]

 By Lemma~\ref{lem:countpartitions} with $n = |\cC_\eps| \leq (1/\eps)^d = T^d$ we have $|\cD_\eps^{(k)}| \leq (T^dk^2)^{k(d+1)}$. Therefore,
\[D \leq (2d)^k (T^dk^2)^{k(d+1)}.\]
Taking logs and setting $\delta = \frac{1}{T^{2d}}$,
\[\ln\l(\frac{D}{\delta}\r) \leq k\ln(2d) + dk(d+1)\ln(T) + 2k(d+1)\ln(k) + 2d\ln(T) = O(kd^2 \ln(T)).\]
Substituting this into Theorem~\ref{calibeating_main} yields, with probability at least $1 - \frac{1}{T^{2d}}$,
\[MDSR(k) =O\l(L\sqrt{T(kd^2 \ln T)}\r) = O\l( d L \sqrt{kT \ln(T)}\r)\]

%% file: 6lowerbounds.tex
\section{Relating Downstream Swap Regret to Calibration Error}\label{sec:5lowerbounds}

In follow-up work to \cite{rothshi}, \cite{CDL-huWu} showed that in one dimension it is possible to guarantee $\Tilde{O}(\sqrt{T})$ swap regret for all downstream agents without any dependence on the number of actions. This motivated an open question: is a dependence on the number of actions necessary in higher dimensions, or is it removable? In this section we give evidence to suggest it is necessary, by showing that maximal downstream swap regret of order $O_d\l(T^{1-\delta}\r)$ for any $\delta > 0$ --- with the implied constant depending on the dimension but not on the number of actions --- would imply a breakthrough result on high-dimensional calibration, namely a sequential forecasting algorithm over the $d$-dimensional simplex with Expected Calibration Error $O_d\l(T^{1-\delta/2}\r)$.
Our proof leverages $d$-dimensional generalizations of the well-known observations, in dimension $d=1$, 
that average squared $\ell_2$-calibration error upper-bounds the square of Expected Calibration Error 
(e.g., Lemma 3.1.1 of~\cite{roth2026uncertain}) 
and that squared $\ell_2$-calibration error can be interpreted as swap regret 
for a decision problem with the squared-loss objective
(e.g., Theorem 14 of~\cite{Ucal-kleinberg}).

\begin{definition}
    For a transcript $\pi_T = (\hat{y}_{1:T}, y_{1:T})$, let $P = \{\hat{y}_t: t \in [T]\}$ be the set of predicted values. For a prediction $\hat{y} \in P$ we define its \emph{count} $n(\hat{y}) \coloneqq |\{t \in [T]: \hat{y}_t = \hat{y}\}|$ and \emph{recalibrated prediction}
    $\sigma(\hat{y}) \coloneqq \frac{\sum_{t: \hat{y}_t = \hat{y}} y_t}{n(\hat{y})}$ and let $Q$ be the set of recalibrated predictions $\{\sigma(\hat{y}) : \hat{y}\in P\}$. The 
    $\ell_1$-calibration error (ECE) is \[ECE(\pi_T) \coloneqq \sum_{\hat{y} \in P} n(\hat{y}) \|\hat{y} - \sigma(\hat{y})\|_1 = \sum_{t =1}^T \|\hat{y}_t - \sigma(\hat{y}_t)\|_1\]
    The squared $\ell_2$ bias (a.k.a.~$\ell_2$ calibration error) is 
    \[K_2(\pi_T) \coloneqq \sum_{\hat{y} \in P} n(\hat{y}) \|\hat{y} - \sigma(\hat{y})\|_2^2 = \sum_{t =1}^T \|\hat{y}_t - \sigma(\hat{y}_t)\|_2^2\]
\end{definition}

The $d=1$ case of the following lemma is Lemma 3.1.1 of~\cite{roth2026uncertain}.
\begin{lemma}
    Consider a prediction space $\cC \subseteq [0,1]^d$. For any transcript $\pi_T$, \[\frac{1}{d} \l(\frac{ECE(\pi_T)}{T}\r)^2 \leq \frac{K_2(\pi_T)}{T} \leq \frac{ECE(\pi_T)}{T}\]
\end{lemma}
\begin{proof}
    Note that $\hat{y} - \sigma(\hat{y}) \in [-1,1]^d$ so $\|\hat{y} - \sigma(\hat{y})\|_2^2 \leq \|\hat{y} - \sigma(\hat{y})\|_1$ from which the right hand inequality follows immediately. For the left inequality apply Cauchy-Schwarz,
    \[\l(\sum_{t = 1}^T \|\hat{y} - \sigma(\hat{y})\|_1\r)^2 \leq T\sum_{t = 1}^T \|\hat{y} - \sigma(\hat{y})\|_1^2\]
    Using norm equivalence, $\|\hat{y} - \sigma(\hat{y})\|_1^2 \leq d\|\hat{y} - \sigma(\hat{y})\|_2^2$. Dividing by $T^2$ completes the proof.
\end{proof}
The $d=1$ case of the following lemma is Theorem 14 of~\cite{Ucal-kleinberg}.
\begin{lemma}\label{lem:lower}
    Consider the prediction space $\cC = \Delta_{d} = \{y \in \R_{\geq 0}^d \mid \sum_{i \in [d]} y_i = 1\}$. For any transcript $\pi_T =(\hat{y}_{1:T}, y_{1:T})$ there exists a downstream agent with at most $2T$ actions with swap regret exactly $\frac{1}{2} K_2$. Consequently, \[\frac{1}{2d}\l(\frac{ECE(\pi_T)}{T}\r)^2 \leq \frac{1}{2}\l(\frac{K_2(\pi_T)}{T}\r) \leq \l(\frac{MDSR_{\pi_T}(2T)}{T}\r) .\]
\end{lemma}
\begin{proof}
Define an action set $\cA$ and utility function $u: \cA \times \cC \to [0,1]$ as follows: for every $v \in Q \cup P$, include an action $a_{v}$ with utility \[u(a_{v}, y) \coloneq \langle v, y\rangle - \frac{1}{2}\|v\|_2^2 + \frac{1}{2}.\] Note that this utility is affine in $y$. Moreover, for $v, y \in \Delta_{d}$, $\langle v, y\rangle \in [0,1]$ and $ \|v\|_2^2 \in [0,1]$, hence $u(a_v, y) \in [0,1]$ as required. For a fixed $y$, $\langle v, y\rangle - \frac{1}{2}||v||_2^2 + \frac{1}{2}$ is maximized at $v = y$. Consequently, for any $\hat{y}_t \in Q \cap P$, $a_{\hat{y}_t} = \br(u, \hat{y}_t)$. Now we can compute
    \begin{align*}
        DSR_{\pi_T}(u) &= \max_{\phi: \cA \to \cA} \sum_{t=1}^T u(\phi(a_{\hat{y}_t}), y_t) - u(a_{\hat{y}_t}, y_t)\\
        &= \max_{\phi: \cA \to \cA} \sum_{\hat{y} \in P} \sum_{t: \hat{y_t = \hat{y}}} \l[u(\phi(a_{\hat{y}}), y_t) - u(a_{\hat{y}}, y_t)\r]\\
        &= \max_{\phi: \cA \to \cA} \sum_{\hat{y} \in P} n\l(\hat{y}\r) \l[u(\phi(a_{\hat{y}}), \sigma(\hat{y})) - u(a_{\hat{y}}, \sigma(\hat{y}))\r] \quad \text{(utility is affine)}\\
        &=  \sum_{\hat{y} \in P} n\l(\hat{y}\r) \l[u(a_{\sigma(\hat{y})}), \sigma(\hat{y})) - u(a_{\hat{y}}, \sigma(\hat{y}))\r]\\
        &=  \sum_{\hat{y} \in P} n(\hat{y}) \l[\l(\langle\sigma(\hat{y}), \sigma(\hat{y})\rangle - \frac{1}{2}||\sigma(\hat{y})||_2^2\r) - \l(\langle\hat{y}, \sigma(\hat{y})\rangle - \frac{1}{2}||\hat{y}||_2^2\r)\r]\\
        &= \sum_{\hat{y} \in P} n(\hat{y}) \frac{1}{2}\|\hat{y} - \sigma(\hat{y})\|_2^2 = \frac{K_2(\pi_T)}{2}
    \end{align*}
\end{proof}
\begin{corollary} \label{cor:ece-breakthru}
  Suppose a forecasting algorithm produces transcripts $\pi_T$ satisfying the guarantee 
  $\mathbb{E} \l[ MDSR_{\pi_T}(2T) \r] \leq C_d T^{1-\delta}$, then also
  $\mathbb{E} \l[ ECE(\pi_T) \r] \leq \sqrt{2 d C_d} \cdot T^{1 - \delta/2}. $ 
\end{corollary}
Since the strongest Expected Calibration Error guarantee known for fixed dimension
$d \gg 1$ and $T \to \infty$ is $T^{1 - O(1/d)}$, \Cref{cor:ece-breakthru} implies
that a downstream swap regret bound of the form $C_d T^{1-\delta}$ for constant
$\delta > 0$ and $T \to \infty$, independent of the number of actions, would 
imply a breakthrough on Expected Calibration Error for high-dimensional forecasting.

Moreover, by \Cref{lem:lower}, a guarantee of the form $\mathbb{E} \l[ MDSR_{\pi_T}(k) \r] = \tilde{O}_d(\sqrt{T}k^{\gamma})$ for every $k$ implies $\mathbb{E} \l[ ECE(\pi_T) \r]  = \tilde{O}_d(T^{3/4 + \gamma/2})$. In particular, a downstream swap regret guarantee with subpolynomial dependence on the number of actions would imply an expected calibration error guarantee of $O(T^{3/4 + o(1)})$ in the corresponding dimension. Since no guarantee of this form is known for $d > 2$, this suggests that extending the $\Tilde{O}(\sqrt T)$, action-independent downstream swap regret guarantees of \cite{CDL-huWu} even to low dimensional settings requires new progress on multidimensional expected calibration error.

\section{Discussion and Conclusion}
We have given forecasting algorithms that guarantee low swap regret for any downstream agent who best responds to the forecasts. In particular, 
for two dimensional outcome spaces (e.g. the simplex over 3 classes) we give a polynomial-time algorithm that guarantees $\Tilde{O}(\sqrt{kT})$ swap-regret for all downstream agents. This improves on prior work both in terms of the swap regret rate and computational efficiency. For general dimension $d$, we give an inefficient algorithm that guarantees $\Tilde{O}(d\sqrt{kT})$ swap regret. Our bounds provide simultaneous guarantees for all decision makers without a huge statistical overhead; they match the dependence for single-agent swap regret algorithms on the time horizon $(T)$ and size of action set $(k)$  up to logarithmic factors.
The main improvements came from exploiting a sharper geometric characterization of best response regions. 

Several questions remain open. First, although our low-dimensional algorithm is computationally efficient,
one of the central remaining questions is computational efficiency in high dimensions. Our high-dimensional algorithm relies on an exhaustive enumeration of an exponential number of best-response partitions. It would be interesting to develop an oracle-efficient version of the algorithm, replacing the explicit enumeration with an appropriate learning oracle. A second question is whether the dependence on the number of actions
can be improved in low dimensions. The connection to calibration in \Cref{sec:5lowerbounds} suggests that removing a polynomial dependence altogether would require advancements in multidimensional calibration, but do not rule out intermediate improvements such as replacing the $\sqrt{k}$ by a smaller polynomial factor. Finally, it would be natural to extend the algorithms to contextual forecasting settings, where predictions and outcomes are made conditional on observed features.

%% file: 7apdx.tex
\appendix
\section{Appendix}
\subsection{High Probability Martingale Bound}
To convert the guarantee of \cite{noarov2023high} from expectation to high probability over the realized forecasts, we apply the following version of Freedman's inequality \cite{3d1c033e-51cd-3f22-8e4e-611c9a4badb9}.
\begin{lemma}[Lemma 3 in \cite{kakade2008generalization}]
    Suppose $X_1,\ldots,X_T$ is a martingale difference sequence with $|X_t|\le b$. Let
\[
\text{Var}_t X_t \;=\; \text{Var}\!\left(X_t \mid X_{t-1}, \ldots, X_1\right).
\]
Let $V=\sum_{t=1}^T \text{Var}_t [X_t]$ be the sum of conditional variances of $X_t$'s. Further, let $\sigma=\sqrt{V}$. Then we have, for any $\delta<1/e$ and $T\ge 3$,
\[
\operatorname{Pr}\!\left(
\sum_{t=1}^T X_t \;>\; \max\left\{2\sigma,\; 3b\sqrt{\ln(1/\delta)}\right\}\sqrt{\ln(1/\delta)}
\right)\;\le\; 4\ln(T)\,\delta .
\]
\end{lemma}
The conditioning on $X_1, \ldots, X_{t-1}$ can be replaced by a conditioning on the filtration $\cF_{t-1}$ in the proof. Moreover, note that $\text{Var}_t X_t = \mathbb{E}[X_t^2|\cF_{t-1}]$ because $X_t$ is a martingale difference sequence. Furthermore, by considering the negation of the sequence we get the following two-sided bound. 
\begin{lemma}\label{lem:freedmanmassaged}
    Suppose $X_1,\ldots,X_T$ is a martingale difference sequence with filtration $\{\cF_t\}_{t=0}^T$ and $|X_t|\le b$.
Let \[V=\sum_{t=1}^T \mathbb{E}[X_t^2|\cF_{t-1}]\] be the sum of conditional variances of $X_t$'s. Then we have, for any $\delta<1/e$ and $T\ge 3$,
\[
\operatorname{Pr}\!\left(
\l|\sum_{t=1}^T X_t\r| \;>\; 2\sqrt{V\ln(1/\delta)} + 3b\ln(1/\delta)\right) \leq 8 \ln(T)\delta
\]
\end{lemma}
Using this bound, we can now prove the high probability guarantee on the algorithm of \cite{noarov2023high} from \Cref{sec:4lowdims}.

\uniformwhp*
\begin{proof}\label{uniformwhp_proof}
    By Theorem~\ref{thm:nor}, \[\mathbb{E}_{\hat{y}'_t \sim p_t} \l\|\sum_{t=1}^T\left[1[\hat{y}'_t \in S](\hat{y}'_t - y_t)\right]\r\|_\infty = O\left(\ln(d|\cE|T) + \sqrt{\ln(d|\cE|T)n_T(S)} + \eps T\right),\]
    so it suffices to show that with probability $1 - \delta$, the following holds for every $S \in \cE$
    \[\l\|\sum_{t=1}^T \l(1[\hat{y}_t \in S](\hat{y}_t - y_t) - \mathbb{E}_{\hat{y}'_t \sim p_t} 1[\hat{y}'_t \in S](\hat{y}'_t - y_t) \r)\r\|_{\infty} \leq O\left(\ln(d|\cE|T/\delta) + \sqrt{\ln(d|\cE|T/\delta)n_T(S)} + \eps T\right).\]
    Noting that
    \[\mathbb{E}_{\hat{y}_t \sim p_t}[1[\hat{y}_t \in S](\hat{y}_t - y_t) - \mathbb{E}_{\hat{y}'_t \sim p_t} 1[\hat{y}_t' \in S](\hat{y}_t' - y_t)] = 0,\] the restriction of $D_t(S) = 1[\hat{y}_t \in S](\hat{y}_t - y_t) - \mathbb{E}_{\hat{y}'_t \sim p_t} 1[\hat{y}_t' \in S](\hat{y}_t' - y_t)$ to any dimension $j$ forms a martingale difference sequence with filtration $\cF_{t-1} = \{p_{\leq t}, \hat{y}_{<t}, y_{<t}\}$. We will bound the martingale by controlling the deviation at each step and applying Freedman's inequality. Note that $\|D_t(S)\|_\infty \leq 2$ and $\mathbb{E}\l[\|D_t(S)\|_\infty\r | \cF_{t-1}] \leq 2 \mathbb{E}_{\hat{y}'_t \sim p_t}[1[\hat{y}'_t \in S]] $ and hence 
    \[\mathbb{E}[D_{t, j}(S)^2 | \cF_{t-1}] \leq 2\mathbb{E}[|D_{t, j}(S)| \mid \cF_{t-1}] \leq 4 \mathbb{E}_{\hat{y}'_t \sim p_t} [1[\hat{y}'_t \in S]].\]
    Therefore \[\sum_{t=1}^T \mathbb{E}[D_{t, j}(S)^2 | \cF_{t-1}] \leq \sum_{t=1}^T 4 \mathbb{E}_{\hat{y}'_t \sim p_t} [1[\hat{y}'_t \in S]] =  4n_T(S).\]
Applying Lemma~\ref{lem:freedmanmassaged} to $D_{t,j}(E)$ with $b=2$ gives for $\delta' < 1/e$ and $T \geq 3$
\[
\mathbb{P}\left(
\left|\sum_{t=1}^T D_{t,j}(S)\right| \geq 2\sqrt{4n_T(S) \ln(1 / \delta') } + 6\ln(1 / \delta')
\right)
\leq8\ln(T)\delta'.
\]
Taking a union bound over $j \in [d]$ and $S \in \cE$ yields
\[
\mathbb{P}\left(\exists {S \in \cE},
\left\|\sum_{t=1}^T D_{t}(S)\right\|_{\infty} \geq 2\sqrt{4n_T(S) \ln(1 / \delta') } + 6\ln(1 / \delta')
\right)
\leq 16d|\cE|\ln(T) \delta'.
\]
Setting $\delta' = \frac{\delta}{16d|\cE|\ln(T)} < 1/ e$ gives
\[
\mathbb{P}\left(
\exists{S \in \cE},
\left\|\sum_{t=1}^T D_{t}(S)\right\|_{\infty} \geq 2\sqrt{4n_T(S) \ln(16d|\cE|\ln(T)/ \delta) } + 6\ln(16d|\cE|\ln(T) / \delta)\right) \leq \delta
\]
hence, with probability at least $1 - \delta$, the following holds for every $S \in \cE$
    \[\l\|\sum_{t=1}^T \l(1[\hat{y}_t \in S](\hat{y}_t - y_t) - \mathbb{E}_{\hat{y}'_t \sim p_t} 1[\hat{y}_t' \in S](\hat{y}_t' - y_t) \r)\r\|_{\infty} \leq O\left(\ln(d|\cE|T/\delta) + \sqrt{\ln(d|\cE|T/\delta)n_T(S)} + \eps T\right),\]
which completes the proof.
\end{proof}

\input{convex_polygon_lemmas}
\input{5highdims}
\input{8Calibeating}

\subsection{Lipschitz Constant}\label{lipschitz}
Here we show the Lipschitz constant $L$ can always be assumed to be at most $d\sqrt{d}$ by performing an affine transformation of the prediction space. The proof follows by transforming the John ellipsoid to a maximal euclidean ball within $[0,1]^d$. The transformed space will contain a euclidean ball with diameter $1/d$. Then we show that an affine function on such a ball bounded in $[0,1]$ is $d\sqrt{d}$ Lipschitz in the $\ell_\infty$ norm.

Core to the proof is the famous result of \cite{john1948extremum}, any convex body can be sandwiched between two ellipsoids which differ by dilation factor $d$. We give a common formulation here.

\begin{lemma}[John Ellipsoid \cite{john1948extremum}]
    For any compact, convex body $K \in \R^d$ there exists an ellipsoid $E$ such that $K \subset E$ while $E' \subset K$ where $E$ is a dilation of $E'$ about its center by a factor of $d$.
\end{lemma}

We use this to prove the following Lemma.
\begin{lemma}
    Let $\cC \subset \R^d$ be convex and compact. There exists an affine bijection $\Phi: \cC \to \cC' \subseteq [0,1]^d$ such that for any affine utility $u(a, \cdot): \cC \to [0,1]$, the transformed utility $u'(a, y) \coloneq u(a, \Phi^{-1}(y))$ is $L$-Lipschitz in the $\ell_\infty$ norm with $L \leq d\sqrt{d}$. Moreover if $\cC' = [0,1]^d$ then $L \leq 1$.
\end{lemma}
\begin{proof}
    On an axis aligned hypercube with side length $s$, any affine function bounded in $[0,1]$ is $1 / s$-Lipschitz in the $\ell_\infty$ norm. This immediately gives the bound $L \leq 1$ for $\cC' = [0,1]^d$. For the general claim, let $\Phi$ be the affine transformation mapping the John's Ellipsoid $E$ containing $\cC$ to the Euclidean ball of norm $\frac{1}{2}$ positioned at the center of $[0,1]^d$. Then $\cC'$ contains a ball of diameter $1/d$ with the same center, which in turn contains a cube with side length $\frac{1}{d\sqrt{d}}$. Hence, $L \leq d\sqrt{d}.$
\end{proof}
\subsection{AI disclosure}\label{app:sec:ai_disclosure}
We used Gemini
to develop our arguments for the pulling triangulation algorithm. 
Refine.ink was used to check this paper for consistency and clarity. 

%% file: convex_polygon_lemmas.tex
\subsection{Triangulation Lemmas}\label{triangulationlemmas}
We give here the full proof of Lemma~\ref{lem:triangulation} from section~\ref{sec:4lowdims}. First, we give the following geometric lemma bounding the number of simplices needed to triangulate a polytope.

\begin{lemma}[Triangulation Size]\label{lem:tri_bound}
    A simple polytope $P\subseteq \R^d$ with $k$ facets can be triangulated using at most $2d! \binom{k}{\lfloor d/2 \rfloor} = O\left(d! k^{\lfloor d/2 \rfloor}\right)$ simplices and every vertex of these simplices will also be a vertex of $P$.
\end{lemma}
\begin{proof}
    We establish some notation.
    The collection of $i$-dimensional faces on a polytope $Q\subseteq \R^d$ is $\mathcal{F}_i(Q)$ and the number of these faces is $f_i(Q) = |\mathcal{F}_i(Q)|$.
    Note that $f_{d-1}(P) = k$ by assumption.
    Finally, $\dim(F)$ is the dimension of face $F$.
    
    We first bound the number of vertices $f_0(P)$ on $P$ using the polar/dual polytope and the upper bound Theorem \cite[Section 5.5]{matousek2013lectures}.
    \begin{align*}
        f_0(P)
        &= f_{d-1}(P^{\circ}) 
        &&\text{duality} \\
        &\le 2 \binom{f_0(P^{\circ})}{\floor{d/2}}
        &&\text{upper bound theorem}\\
        &= 2\binom{k}{\floor{d/2}}
        &&f_0(P^\circ) = f_{d-1}(P) = k
    \end{align*}

    The proof is completed by showing that the number of simplices required to triangulate $P$ is at most $d! f_0(P)$.
    We specifically consider the size of the triangulation generated by the pulling triangulation algorithm $\text{Pull}$, formalized below, on an arbitrary ordering of the vertices $v_1,\dots,v_m$.
    We use the notation $v_F = \arg\min\{\mathcal{F}_0(P) \cap F\}$ to denote the vertex $v$ earliest in the fixed ordering that is a member of face $F$.
    This algorithm is defined recursively; it takes as input a face $F$ (or the polytope $P$ itself).
    At a high level, \textsc{Pull} identifies the vertex earliest in the ordering from its input face $v_F$ and considers facets of the input polytope that do not include that vertex.
    It recursively triangulates these facets using lower-dimensional simplices, which are extended to higher dimensions by the addition of $v_F$ as a vertex.
    \begin{align*}
        &\text{If }\dim(F)=0
        &&\textsc{Pull}(F)
        = F \\
        &\text{Otherwise}
        &&\textsc{Pull}(F) = \bigcup_{\substack{F' \in \mathcal{F}_{\dim(F)-1}(F) \\ v_F \notin F'}} \{\textsc{ConvHull}(v_F, S) \mid S\in \textsc{Pull}(F')\}
    \end{align*}
    We argue that $\textsc{Pull}(P)$ outputs a triangulation of $P$ with at most $d! f_0(P)$ simplices.
    It is known that \textsc{Pull} yields a triangulation \cite[Section 16.2.1]{lee2017subdivisions}.
    Moreover, the vertices of each simplex in the \textsc{Pull} triangulation is also a vertex of $P$.
    
    Notice that each leaf in the recursion tree of \textsc{Pull} corresponds to a simplex in the constructed triangulation.
    We complete the proof by counting the number of root-to-leaf paths in the recursion tree.
    Fix the inputs used in such a path $P,  F_{d-1}, F_{d-2},\dots, F_0$, where $\dim(F_i) = i$.
    Notice that each path has a unique sequence of inputs, and these inputs form a maximal chain of faces $F_{d-1} \supseteq \dots\supseteq F_0$ by construction.
    Each $i$-dimensional face of a simple polytope is contained in $d - i$ faces of $i+1$ dimensions (this is verifiable using the fact that the dual/polar of a simple polytope is simplicial: each face is a simplex) \cite[Section 5.3]{matousek2013lectures}.
    Each vertex is contained in $d$ edges, each of which is contained in $d-1$ faces of $2$-dimensions, and so on.
    Thus, the number of maximal chains containing each vertex is $d!$.
    Using our argument that each simplex in the $\textsc{Pull}$ triangulation is some unique maximal chain, we bound the number of simplices in the triangulation by summing, over each vertex, the number of maximal chains containing the vertex:
    \begin{align*}
        |\textsc{Pull}(P)| \le d! f_{0}(P) \le 2 d! \binom{k}{\floor{d/2}} = O(d! k^{\floor{d/2}})
    \end{align*}
\end{proof}

Using this, we show that best response regions can collectively be decomposed into a well-controlled number of simplices.

\triangulation*
\begin{proof}
    The proof proceeds in four steps: First, we ensure every point of $P$ has a positive margin from the boundary of its best response region. To do this we perturb each utility by adding an action-dependent constant and making sure to preserve best responses for outcomes in $P$. Second, we lift to $\R^{d+1}$ by viewing each action's utility $u(a, \cdot)$ as the hyperplane $z = u(a, y)$; best response regions correspond to portions of these hyperplanes that appear on the boundary. Then, we perform a small random perturbation to ensure that the upper envelope is a simple polytope and triangulate it using \Cref{lem:tri_bound}. Finally, we project back down to the best response regions in $\R^d$ and perform another small perturbation to ensure that all of the points in $P$ lie strictly in the interiors.
    
    Define
    \[\Delta = \min_{\mathclap{\substack{y \in P\\ i \neq j \in [k]}}}\Big\{|u(i, y) - u(j, y)|\; \Big\vert\; u(i, y) \neq u(j, y)\Big\}\]
    and in case the set is empty, define $\Delta \coloneqq 0.1$. By construction $\Delta>0$. Choose constants $c_j = \frac{k+1-j}{k+1}\Delta$ for $j\in[k]$ and draw independent $\eta_j \sim \mathrm{Unif}\!\left(0,\frac{\Delta}{10(k+1)}\right)$. Now, define the utility function $v$ with
\[
v(i,y)\coloneq u(i,y) + c_i + \eta_i.
\]

    Observe that for any point $y \in P$ we will have:
    \[v(i, y) - v(i, y) = u(i, y) - u(j, y) + c_i - c_j + \eta_i - \eta_j\]
    Thus, if $u(i, y) = u(j, y)$ for $i < j$ then $v(i, y) > v(j, y)$. On the other hand, if for any $i\neq j$ we have $u(i, p) \neq u(j, p)$ then whichever action dominated still does so under $v$ (i.e. $\mathrm{sgn}(v(i, y) - v(j, y)) = \mathrm{sgn}(u(i, y) - u(j, y))$, since $|(c_i-c_j)+(\eta_i-\eta_j)| < \Delta \leq |u(i, y) - u(j, y)|.$

    Now, we triangulate $[0,1]^d$ by considering the upper envelope $G$ (intersected with $[-1,2]^{d+1}$, leaving excess margin on all sides of best response regions) of the function $g_v : [0,1]^d \ni y \to \max_{i \in [k]}v(i, y)$, where we linearly extend $v$ to $[0,1]^d$. The projections of the facets of this function onto $[0,1]^d$ partition it to the agent's best response regions. We argue that $G$ is simple with probability 1: this is because conditional on the perturbation values of $d+1$ actions, the probability that one other action attains the same utility at that unique point is zero (as it fixes its perturbation). Therefore, there exists a perturbation such that the "lifted polytope" $G$ is simple (and the best response partitions of $P$ are as needed).

    The simple convex polytope $G$ has at most $k+2d +2$ facets by definition (we need to account for the faces of the cube of the $d+1$-dimensional cube). Therefore, by \Cref{lem:tri_bound} it can be triangulated into $2(d+1)!\binom{k+2d + 2}{\lfloor (d+1)/2 \rfloor} = O_d(k^{\lceil d/2 \rceil})$ many simplices.

    A triangulation of $G$ induces a triangulation of each facet $F\subseteq \partial G$ as follows: consider all $(d+1)$-simplices $\tau$ in the triangulation that intersect $F$, and take their intersections $\tau\cap F$. Each set $\tau\cap F$ is a simplex of at most $d$ dimensions, their interiors (if they are non-empty) are disjoint, and their union is $F$. Moreover, it is well known that each original $(d+1)$-simplex $\tau$ has exactly $2^{d+2}-1$ faces. Hence $\tau$ can contribute boundary simplices (i.e., faces contained in $\partial G$) to at most $2^{d+2}-1$ facets of $G$. Therefore, summing over all facets, the total number of simplices appearing in the induced face triangulation is at most $2^{d+2}-1$ times the number of $(d+1)$-simplices in the triangulation of $G$. 

    Restrict the triangulation of $G$ to the induced face simplices that lie on the facets of the actions (i.e., on facets of the form $z=v(j,y)$ where $v(j,y)=g_v(y)$), and project each such simplex to $\R^d$ via $(y,z)\mapsto y$. The images have disjoint interiors. Moreover, each point in $P$ has distance at least $\Delta/ (10(k+1)dL)$ the boundary of the union of simplices that correspond to its best response region. To satisfy (2), (3), and (4) we can perform a small random translation of the entire complex of simplices of distance at most $\Delta/(100(k+1)dL)$.
    Specifically, translate each simplex in the triangulation by the same vector $\epsilon w$, where
    $\epsilon \in \mathrm{Unif}\left(0, \frac{\Delta}{100(k+1)dL}\right)$ and $w$ is drawn uniformly at random from the boundary of the $d$ dimensional unit sphere.
    Since the faces of each simplex have volume zero, no point in $P$ will lie on any face almost surely.
\end{proof}

%% file: 5highdims.tex
\subsection{Alternative Approach to  High-Dimensional Forecasting}\label{OldHighDim}
Here we give an alternate approach for high dimensional forecasting maintaining unbiasedness on every subset of the discretized prediction set that arises from the intersection of a bounded number of halfspaces. This bound is strictly worse than the bound from \cref{sec:3highdims} (by a factor of $\sqrt{k}$) and is included for expository purposes. 

\begin{theorem}\label{thm:highdim}
    Consider a prediction space $\mathcal{C} \subseteq [0,1]^d$. Let $\mathcal{B}_\epsilon^{(k)} = \{C_\epsilon \cap B : B = \bigcap_{i \in [k]} H_i; H_i \subset \mathbb{R}^d \text{ is a halfspace}\}$. Using the Unbiased Predictions algorithm of \cite{noarov2023high} with $\mathcal{E} =  \mathcal{B}_\epsilon^{(k)}$ and $\epsilon = \frac{1}{T^2}$, then for any $\delta > 0$, with probability at least $1-\delta$,
    \[MDSR_{\pi_T}(k) = O(kdL\sqrt{T\ln(T/\delta)}).\] 
    Therefore, in expectation, the maximum swap regret across all possible bounded utility agents with at most $k$ actions is $O\l(kdL\sqrt{T\ln(T)}\r)$.
\end{theorem}
\begin{proof}
    Recall that $|\cC_\eps| \leq (1/\eps)^d = T^{2d}$. By \Cref{lem:simplex_count}, the number of events $|\cE| = |\cC_\eps|^{O(dk)} = T^{O(d^2k)}$. Applying \Cref{cor:jensen}, with probability at least $1-\delta$, for any set of disjoint events $\cE' \subseteq \cE$, \[\sum_{E \in \cE'} \beta(E) \leq c'\l(\sqrt{|\cE'|T\ln(d|\cE|T / \delta)}\r) = O\l(d \sqrt{|\cE'|Tk\ln(T / \delta)}\r).\]  
   Recall the early observation that every best response region of any agent with at most $k$ actions is the intersection of the prediction space with at most $k-1$ half-space and therefore every best response event of every such agent is in $\cE$. Therefore, taking $\cE'$ to be the best response events of some agent, 
    \[\sum_{a \in \cA} \beta(E_{u,a}) = \sum_{E \in \cE'} \beta(E) \leq O\l(d\sqrt{k^2T\ln(T/\delta)}\r) = O\l(dk\sqrt{T\ln(T/\delta)}\r)\]
\end{proof}

%% file: 8Calibeating.tex
\subsection{Adapting Multiobjective Optimization}
\label{calibeatingmodel}
The multiobjective optimization framework of \cite{calibeating} considers a more general setting than stated in the body of our paper. We describe here the original setting (described in  Appendix A.3 of \cite{calibeating}) and how our simplified setting relates. 

The setup is a two-player sequential game between a Learner and an Adversary. The game plays out over rounds $t = 1, \ldots, T$. Each round, the Adversary publishes the following: (i) the Learner's pure action set $\cA^t$, (ii) the Adversary's convex action set $\mathcal{Y}^t$ embedded in euclidean space, and (iii) a $D$ dimensional vector loss function $\ell^t(\cdot,\cdot): \cA^t \times \cY^t \to [-C,C]^D$.
Then, the following sequence occurs:
\begin{enumerate}
    \item The Learner chooses $x_t \in \Delta\cA^t$.
    \item The Adversary outputs $y_t \in \cY^t$.
    \item The Learner outputs $a_t \sim x_t$.
    \item The Learner observers loss $\ell^{t} (a_t, y_t)$
\end{enumerate}
For a given round, the Probabilistic AMF value is defined as \[w_A^{t} = \sup_{y_t \in \cY^t} \min_{\sigma_t \in \Delta \cX^{t}} \max_{j \in D} \mathbb{E}_{x_t \sim \sigma{t}}[l_j^{t}(x_t, y_t)].\]
We use the following guarantee.
\begin{theorem}[\cite{calibeating}, Theorem~A.2]\label{calibeating}
    Fix any $\delta \in (0,1)$. Given $T \geq \ln(D)$, there exists an algorithm that guarantees with probability $1 - \delta$ over the learner's randomness, 
\[\max_{j \in [D]}\l\{\sum_{t = 1}^T l_j^t(x_t, y_t)\r\} \leq 8 C\sqrt{T \ln\l(\frac{D}{\delta}\r)} + \sum_{t=1}^T w_A^t\] 
\end{theorem} 

We make the following notational changes and simplifications in the main body of the paper:
\begin{itemize}
    \item To avoid clashes of notation, we replace their usage of $\cA^t$ and $\Delta\cA^t$ with $\cX^t$ and $\Delta\cX^t$, respectively. Moreover, we assume that the learner and adversary draw pure actions simultaneously, which only improves guarantees for the learner. We denote the learner's pure action on round $t$ as $x_t$ rather than $a_t$.
    \item We fix the actions space $\cX, \cY$ and loss $l$ across time.
    The probabilistic AMF value is the same at every timestep therefore can be noted by a single variable $w$.
    \item We fix the diameter of the decision spaces $C = 1$ for simplicity.
\end{itemize}